\documentclass[preprint]{elsarticle}
\usepackage{soul, color}
\usepackage{colortbl}
\usepackage{bm}
\usepackage{graphicx}
\usepackage{booktabs}
\usepackage{tikz}
\usepackage{caption}
\usepackage{arydshln}
\usepackage{subcaption}
\usepackage{amsmath}
\usepackage{amssymb}
\usepackage{amsthm}
\usepackage{dsfont}
\usepackage{bbold}
\usepackage{listings}
\usepackage{setspace}
\usepackage{url}
\usepackage{verbatim}
\usetikzlibrary{positioning,decorations.pathreplacing,shapes}

\newcommand{\var}{\operatorname{Var}}

\newcommand{\R}{\mathbb{R}}

\definecolor{codegreen}{rgb}{0,0.6,0}
\definecolor{codegray}{rgb}{0.5,0.5,0.5}
\definecolor{codepurple}{rgb}{0.58,0,0.82}
\definecolor{backcolour}{rgb}{0.95,0.95,0.92}
\definecolor{Gray}{gray}{0.8}
\definecolor{lightyellow}{RGB}{245,238,197}
\newcolumntype{a}{>{\columncolor{lightyellow}}c}
\newcolumntype{b}{>{\columncolor{backcolour}}c}
\renewcommand\arraystretch{1.2}
\newcommand{\figwidth}{0.45\textwidth}
\lstdefinestyle{mystyle}{
    backgroundcolor=\color{backcolour},
    commentstyle=\color{codegreen},
    keywordstyle=\color{magenta},
    numberstyle=\tiny\color{codegray},
    stringstyle=\color{codepurple},
    basicstyle=\ttfamily\footnotesize,
    breakatwhitespace=false,
    breaklines=true,
    captionpos=b,
    keepspaces=true,
    numbers=left,
    numbersep=5pt,
    showspaces=false,
    showstringspaces=false,
    showtabs=false,
    tabsize=2
}

\newtheorem{definition}{Definition}
\newtheorem{proposition}{Proposition}

\begin{document}
\begin{frontmatter}
\title{Probabilistic Reconciliation of Count Time Series}
\author[add1]{Giorgio Corani}
\ead{giorgio.corani@idsia.ch}
\author[add1]{Dario Azzimonti}
\ead{dario.azzimonti@idsia.ch}
\author[add1]{Nicolo Rubattu}
\ead{nicolo.rubattu@supsi.ch}
\address[add1]{IDSIA\\ Dalle Molle Institute for Artificial Intelligence\\ USI-SUPSI \\
	CH-6962 Lugano, Switzerland}

\begin{abstract}
	Forecast reconciliation is an important  research topic. 
	Yet, there is currently neither  formal framework nor practical method
	for the  probabilistic reconciliation of count time series.
	In this paper we  propose a  definition of coherency  and reconciled probabilistic forecast  which applies to both real-valued and count variables and a novel method  for probabilistic reconciliation. It is based on a generalization of Bayes' rule and it can reconcile both real-value and count variables. When applied to count variables, it yields a reconciled probability mass function.
	Our experiments with the temporal reconciliation of count variables show  a major forecast improvement
	compared to  the   probabilistic Gaussian reconciliation. 
\end{abstract}
\end{frontmatter}

\section{Introduction}

Time series are often organized into a hierarchy.
For example, the total sales of a product in a country can be divided into regions and then into sub-regions. Forecasts of hierarchical time series should be \textit{coherent}; for instance, the sum of the forecasts of the different regions should be equal to the forecast for the entire country. Forecasts are incoherent if they do not satisfy such constraints. 
\textit{Reconciliation} methods \citep{hyndman2011optimal,wickramasuriya2019optimal}
compute coherent forecasts  by combining the \textit{base forecasts} generated independently for each time series, possibly incorporating non-negativity constraints \citep{wickramasuriya2020non-negative}.
Reconciled forecasts are generally more accurate than the base forecasts;
indeed, forecast reconciliation is related to forecast combination  \citep{hollyman2021understanding, di2022forecast}.

A special case of reconciliation is constituted by
\textit{temporal hierarchies} \citep{athanasopoulos2017_temporal},
which reconcile  base forecasts computed for the same variable at different frequencies  (e.g., monthly, quarterly and yearly); they generally  improve the forecasts \citep{KOURENTZES-elucidate}  of smooth and  intermittent time  series.

As for probabilistic reconciliation,
\cite{panagiotelis2022} proposed a seminal framework which  interprets  reconciliation as a projection. 
Other methods for  probabilistic reconciliation have been proposed \citep{jeon2019probabilistic, corani_reconc, taieb2021hierarchical, rangapuram2021end}, but none of them
reconciles count variables. 

Our first contribution is 
the  definition of coherency and reconciliation   for
count variables.
Then, we propose a novel approach for probabilistic  reconciliation,  based on conditioning.
As a first step, our   method 
computes a joint distribution on the entire hierarchy, using as  source of information
the base forecast of the bottom variables; this is the  \textit{probabilistic bottom-up} reconciliation.
Then it  updates the joint distribution by conditioning on the information contained in the base forecast of the upper variables, 
using the method of \emph{virtual evidence} \citep{pearl1988probabilistic, chan2005revision}.  
Our approach can  reconcile  both real-value and count variables; in this paper however we focus on count variables.
In this case, we obtain a  reconciled probability mass function defined  over counts.
We show extensive experiments on the temporal reconciliation of  count time series, reporting major empirical improvements compared to  probabilistic reconciliation based on  Gaussian  assumptions. 

The paper is organized as follows:
Section~\ref{sec:t-hier} reviews temporal hierarchies; in Section~\ref{sec:probabilistic} we propose a definition of  coherent and reconciled forecasts which applies to both continuous and real-valued variables. In Section~\ref{sec:reconc_virtual_ev} we describe  our reconciliation method
and in Section~\ref{sec:experiments} we present the experimental results. 

\section{Temporal Hierarchies}
\label{sec:t-hier}

Temporal hierarchies \citep{athanasopoulos2017_temporal}
enforce coherence between forecasts generated at different temporal scales. 
For instance the temporal hierarchy of Figure~\ref{fig:hierExample} is built on top
of a quarterly time series observed over $t$ years, with observations $q_1, \ldots, q_{4}, q_{5}, \ldots, q_{4\cdot t}$.
The bottom level of the hierarchy contains quarterly observations grouped in  vectors $\mathcal{Q}_j = [q_k : k \mod 4=j]$, $j=1,\ldots, 4$; the semi-annual observations are grouped in the vectors $\mathcal{S}_j = [ s_\ell: \ell \mod 2 =j ]$ where $s_{2i-1} = q_{4(i-1)+1}+q_{4(i-1)+2}$, $s_{2i} = q_{4(i-1)+3}+q_{4i}$ for $i=1, \ldots, t$; finally the annual observations are grouped as $\mathcal{Y}=[a_1,\ldots, a_t]$, where $a_{i}=s_{2i-1}+s_{2i}$, for $i=1,\ldots,t$.

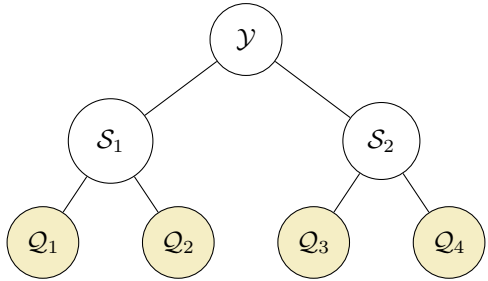
\begin{figure*}[ht!]
	\centering
	\begin{tikzpicture}[level/.style={sibling distance=40mm/#1}, , scale=0.9]
		\node [circle,draw] (z){\textit{$\:\:\mathcal{Y}\:\:$}}
		child {node [circle,draw] (a) {\textcolor{white}{r}$\,\mathcal{S}_1\,$\textcolor{white}{r}}
			child {node [circle,draw, fill=lightyellow] (b) {$\,\mathcal{Q}_{1}\,$}
			}
			child {node [circle,draw, fill=lightyellow] (g) {$\,\mathcal{Q}_{2}\,$}
			}
		}
		child {node [circle,draw] (j) {\textcolor{white}{r}$\mathcal{S}_2$\textcolor{white}{r}}
			child {node [circle,draw, fill=lightyellow] (k) {$\,\mathcal{Q}_{3}\,$}
			}
			child {node [circle,draw, fill=lightyellow] (l) {$\,\mathcal{Q}_{4}\,$}
			}
		};
	\end{tikzpicture}
	\caption{Temporal hierarchy built on top of a quarterly time series.}
	\label{fig:hierExample}
\end{figure*}

We denote by $\mathbf{b}$ the vector of bottom observations, e.g., $\mathbf{b} = [\mathcal{Q}_1^T , \mathcal{Q}_2^T , \mathcal{Q}_3^T, \mathcal{Q}_4^T]^T$, and by $\mathbf{u}$ the vector of upper observations, e.g., $\mathbf{u} = [\mathcal{Y}^T, \mathcal{S}_1^T, \mathcal{S}_2^T]^T$.
We denote by $\mathbf{y}$ 
the vector containing all the observations of the temporal hierarchy, e.g.,
$\mathbf{y}=  [\mathbf{u} \ \mathbf{b}] = [\mathcal{Y}^T, \mathcal{S}_1^T, \mathcal{S}_2^T, \mathcal{Q}_1^T, \mathcal{Q}_2^T , \mathcal{Q}_3^T, \mathcal{Q}_4^T]^T$. 
We denote by  $m$  the number of bottom  observations and by $n$ the total number of observations in the hierarchy.

A hierarchy is characterized by its summing matrix $\mathbf{S} \in \R^{n\times m}$ that defines the relationship between $
\mathbf{b}$ and $\mathbf{y}$, i.e.
\begin{equation*}
	\mathbf{y} = \mathbf{S} \mathbf{b} .
	\label{eq:Reconciliation}
\end{equation*}
The $\mathbf{S}$ matrix of hierarchy in Figure~\ref{fig:hierExample} is:
\newcommand\bigA{\makebox(0,0){\text{\huge{A}}}}
\newcommand\bigI{\makebox(0,0){\text{\huge{I}}}}
\begin{equation*}
	\mathbf{S} = \begin{bmatrix}
		1 & \,\,\, & 1 & \,\,\, & 1 & \,\,\, & 1\\
		1 & \,\,\, & 1 & \,\,\, & 0 & \,\,\, & 0 \\
		0 & \,\,\, & 0 & \,\,\, & 1 & \,\,\, & 1 \\
		1 & \,\,\, & 0 & \,\,\, & 0 & \,\,\, & 0 \\
		0 & \,\,\, & 1 & \,\,\, & 0 & \,\,\, & 0 \\
		0 & \,\,\, & 0 & \,\,\, & 1 & \,\,\, & 0 \\
		0 & \,\,\, & 0 & \,\,\, & 0 & \,\,\, & 1 \\
	\end{bmatrix}
	=
	\begin{bmatrix}
		& \,\,\, &  & \,\,\,   &        & \,\,\, &  \\
		& \,\,\, &  & \bigA   &        & \,\,\, &  \\
		& \,\,\, &  & \,\,\, & \,\,\, &        &  \\  \cdashline{2-6} \
		& \,\,\, &  & \,\,\,   &        & \,\,\, &  \\
		& \,\,\, &  & \bigI   &        & \,\,\, &  \\
		& \,\,\, &  & \,\,\,   &        & \,\,\, &  \\
		& \,\,\, &  & \,\,\,   &        & \,\,\, &
	\end{bmatrix},
	\label{eq:matrixS}
\end{equation*}
\noindent
where  $\mathbf{A} \in \mathbb{R}^{(n-m)\times m}$ encodes which bottom time series should be summed up in order to obtain each upper time series.

\paragraph{Reconciling temporal hierarchies}\label{sec:temporal-hier}
Let us denote by $h$ the forecast horizon, expressed in years.
For instance, $h$=1 implies four
forecasts at the quarterly level, two forecasts at the bi-annual level and one forecast at the yearly level. 

Let us denote by $\hat{\mathbf{u}}_h$ and $\hat{\mathbf{b}}_h$ the \textit{base point forecasts} for the upper and bottom levels of the hierarchy.
The vector of the base point forecasts for the entire hierarchy is $\hat{\mathbf{y}}_h = \begin{bmatrix}
	\hat{\mathbf{u}}_h \\
	\hat{\mathbf{b}}_h
\end{bmatrix}$, 

The base forecast are \textit{incoherent}, i.e., $\hat{\mathbf{y}}_h \neq \mathbf{S}\hat{\mathbf{b}}_h$.
Optimal \textit{reconciliation} methods \citep{hyndman2011optimal,wickramasuriya2019optimal}
adjust the forecast for the bottom level and  sum them up in order to obtain the forecast for the upper levels.
The reconciled  forecasts of the bottom time series and the entire hierarchy are:
\begin{align}
	\tilde{\mathbf{b}}_h &= \mathbf{G} \hat{\mathbf{y}}_h \\
	\tilde{\mathbf{y}}_h & = \mathbf{S}\tilde{\mathbf{b}}_h = \mathbf{SG} \hat{\mathbf{y}}_h.
	\label{eq:s_ytilde}
\end{align}

The core of the minT algorithm \citep{wickramasuriya2019optimal} is
the following expression for  $\mathbf{G}$, which minimizes the mean squared error of  the coherent forecast:
\begin{equation}
	\mathbf{G} = (\mathbf{S}^T \mathbf{W}^{-1} \mathbf{S})^{-1} \mathbf{S}^T \mathbf{W}^{-1}\,,
	\label{eq:PMinTEstimator}
\end{equation}
where $\mathbf{W}$ is the  covariance matrix of the errors of the base forecast. 
The covariance of the reconciled forecasts  is \citep{wickramasuriya2019optimal}
\begin{equation}
	\var(\tilde{\mathbf{y}})=\mathbf{S}  \mathbf{G} \mathbf{W} \mathbf{G}^T \mathbf{S}^T = \mathbf{S}  (\mathbf{S}^T \mathbf{W}^{-1} \mathbf{S})^{-1} \mathbf{S}^T\,.
	\label{eq:varianceMinT}
\end{equation}

In temporal hierarchies,  $\mathbf{W}$ is generally assumed to be diagonal, but it can be defined in different ways.
For instance, \textit{hierarchy variance} \citep{athanasopoulos2017_temporal} adopts the same variances of the base forecasts, allowing heterogeneity within each level of the hierarchy.
For the hierarchy of Fig.\ref{fig:hierExample} it yields:
\begin{equation*}
	\mathbf{W} = \text{diag}(\hat{\sigma}_Y^2, \hat{\sigma}_{S_1}^2, \hat{\sigma}_{S_2}^2,\hat{\sigma}_{Q_1}^2, \ldots, \hat{\sigma}_{Q_4}^2)\,.
	\label{eq:diagW}
\end{equation*}
Instead,  \textit{structural scaling}   \citep{athanasopoulos2017_temporal}
defines $\mathbf{W}$ by assuming: i) the forecasts  of the same level to have the same variance; ii) the variance at each level to be proportional to the number of bottom time series that are summed up in that level. For the hierarchy of Fig.\ref{fig:hierExample}, it yields:
\begin{equation*}
	\mathbf{W} = \text{diag}(4,2,1,1,1,1)\,.
	\label{eq:strucScalingW}
\end{equation*}

\section{Probabilistic Reconciliation}
\label{sec:probabilistic}
The methods discussed so far reconcile the \emph{point forecasts}. 
In the following we review the most important methods for probabilistic reconciliation. 

\citet{jeon2019probabilistic} propose different heuristics (based on minT) for probabilistic reconciliation, one of which is equivalent to reconciling a large number of forecast quantiles.
The algorithm by \cite{taieb2021hierarchical}  yields coherent probabilistic forecasts whose expected value match the mean of  MinT; yet this method does not consider the variance of the base forecast of the upper variables.
\cite{rangapuram2021end} propose a deep neural network model which produces coherent probabilistic forecasts without  any post-processing step, 
by incorporating  reconciliation within a single trainable model.

\cite{corani_reconc} shows that probabilistic reconciliation can be accomplished via Bayes' rule.
First they create a joint predictive distribution for the entire hierarchy, based on the probabilistic base forecast of the bottom time series.
The distribution is then updated in order to accommodate the information contained  in the base forecast of the upper time series.
Under the  Gaussian assumption they obtain the reconciled Gaussian distribution in closed form. The reconciled  mean and variance are  equivalent to those of  minT, despite the different derivation strategy. Also \citep{hollyman2022hierarchies} propose a Bayesian viewpoint of  the reconciliation process.

\citet{panagiotelis2022} proposes a  definition of probabilistic reconciliation
based on projection  and an algorithm  which obtains the reconciled distribution  by minimizing a  scoring rule. However this requires  optimizing via   stochastic gradient descent the  $m \times n$ elements of $\mathbf{G}$, which structurally limits its  scalability.

There is currently no method for the reconciliation of count variables.
To address this problem, we  first
extend to count variables the key definitions of~\citet{panagiotelis2022}.

\subsection{Coherence and reconciliation according to \cite{panagiotelis2022} }
Recalling that  $m$ and $n$ denote  the number of  bottom and total time series,  matrix $\mathbf{S}$ can be seen as a function $s: \mathbb{R}^m \rightarrow \mathbb{R}^n$ which associates to a bottom vector $\mathbf{b}\in \mathbb{R}^m$ the coherent vector  $s(\mathbf{b}) = \mathbf{S}\mathbf{b} \in \mathbb{R}^n$. The $n$-dimensional coherent vectors lie in the vector subspace $\mathfrak{s}$ (spanned by the columns of $\mathbf{S}$), which is well-defined in $\mathbb{R}^n$.
The base forecasts of the bottom time series  can be represented by a probability triple $(\mathbb{R}^m, \mathcal{F}_{\mathbb{R}^m}, \nu)$, where $\mathcal{F}_{\mathbb{R}^m}$ is the (Borel) $\sigma$-algebra associated with $\mathbb{R}^m$. 
\begin{definition}
	\citep{panagiotelis2022} 
	A probability triple $(\mathfrak{s}, \mathcal{F}_{\mathfrak{s}}, \check{\nu}_\mathfrak{s})$ is  coherent with the bottom probability triple $(\mathbb{R}^m, \mathcal{F}_{\mathbb{R}^m}, \nu)$ if: \\
	\begin{equation}
		\check{\nu}_\mathfrak{s}(s(\mathcal{B})) = \nu(\mathcal{B}), \quad \forall \mathcal{B} \in \mathcal{F}_{\mathbb{R}^m},
		\label{eq:panag_coherence}
	\end{equation}
	\label{def:coherence_pan}
\end{definition}
Definition~\ref{def:coherence_pan} implies that incoherent vectors have zero probability under the probability measure $\check{\nu}_{\mathfrak{s}}$.

\subsection{Extension to count variables}
\begin{figure}[!h]
	\centering
	\includegraphics[width=0.80\linewidth]{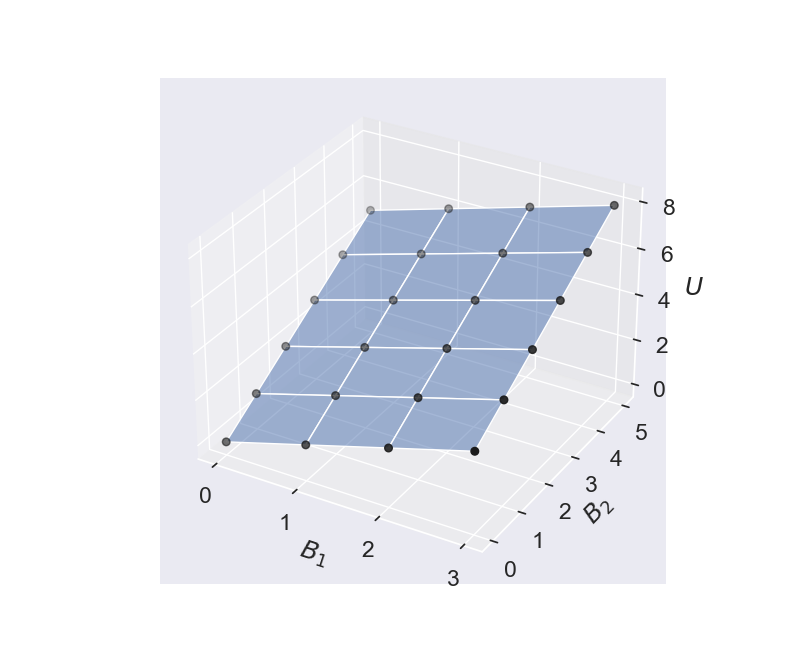}
	\caption{
		$B_1$ and $B_2$ are  two count variables, with
		$U = B_1+B_2$. The black points
		constitute the set of coherent vectors.}
	\label{fig:sB}
\end{figure}

A count variable takes non-negative integer values $\mathbb{N} = \{0, 1, 2, 3, ...\}$.
We denote by $\mathbf{a} \in \mathbb{N}^k$ an array of $k$ non-negative integers which we call vectors, with a slight abuse of notation. 
We have vectors  $\mathbf{y} \in \mathbb{N}^n$, $\mathbf{b} \in \mathbb{N}^m$ and  $\mathbf{u} \in \mathbb{N}^{n-m}$.
The set of coherent vectors in $\mathbb{N}^n$ is:
\begin{equation}
	s(\mathbb{N}^m) = \{ \mathbf{y} \in \mathbb{N}^n : \exists \mathbf{b} \in \mathbb{N}^m \text{ such that } \mathbf{y}= s(\mathbf{b})  \}.
	\label{eq:s(Nm)}
\end{equation}

Eq. \eqref{eq:s(Nm)}  defines the %coherent probability measure  on $s(\mathbb{N}^m)$, 
subset of coherent vectors, see Fig.\ref{fig:sB} for a graphical representation. 
Indeed, no vector subspace can be defined with count variables.
We can now extend to counts the definition of coherence:
\begin{definition}
	A probability triple $(s(\mathcal{X}^m), \mathcal{F}_{s(\mathcal{X}^m)}, \check{\nu})$ is coherent with the bottom probability triple
	$(\mathcal{X}^m, \mathcal{F}_{\mathcal{X}^m}, \nu)$ if:
	\begin{equation}
		\check{\nu}(s(\mathcal{B})) = \nu(\mathcal{B}), \quad \forall \mathcal{B} \in \mathcal{F}_{\mathcal{X}^m},
		\label{eq:coherence}
	\end{equation}
	where $\mathcal{X}^m=\mathbb{R}^m$ in the continuous case and $\mathcal{X}^m= \mathbb{N}^m$ in the discrete case. 
	\label{def:coherence}
\end{definition}
Definitions \ref{def:coherence_pan} and \ref{def:coherence} are  equivalent in the continuous case, as we prove in~\ref{sec:equivalentDefs}.
However, definition \ref{def:coherence} applies also to count variables, in which case  $\check{\nu}$ and $\nu$ are discrete probability distributions. Denoting by $\check{p}$ and $p$ their  probability mass functions, we can write Eq.~\eqref{eq:coherence} as:
\begin{equation}
	\check{p}(b_1, \ldots, b_m, \mathbf{u} = \mathbf{A}\mathbf{b}) = p(b_1, \ldots, b_m),    \qquad \text{ for all } b_1, \ldots, b_m \in \mathbb{N} 
	\label{eq:coherence_discrete}
\end{equation}
Equation \eqref{eq:coherence_discrete} assigns zero probability to  vectors of counts  which are incoherent.

\begin{definition}
	Consider a probabilistic base forecast for $\mathbf{y}$, constituted by the probability triple $(\mathcal{X}^n, \mathcal{F}_{\mathcal{X}^n}, \hat{\nu})$. A reconciled probability distribution $\tilde{\nu}$ is a transformation of the forecast probability measure $\hat{\nu}$ which is coherent and defined on $\mathcal{F}_{s(\mathcal{X}^m)}$. 
	\label{def:reconciled_meas}
\end{definition}

Our definition of reconciled probability triple is more general than that of ~\citet{panagiotelis2022}, which  requires a reconciliation function, 
defined in their algorithm as an affine map. 
This can lead to negatively reconciled forecasts, which is not admissible with  count variables.

\section{Reconciliation based on virtual evidence}\label{sec:reconc_virtual_ev}
We assume the forecasts over counts to be constituted by probability mass functions (pmfs). 
We denote by $\hat{p}$ and  $p_{BU}$ the pmf of the base forecasts and the probabilistic bottom-up respectively. The reconciled pmf is denoted by $\tilde{p}$.

The first step of our  algorithm is to create a joint  pmf for the entire hierarchy using $\hat{p}(\mathbf{b})$ as  the only source of information; this is the \textit{probabilistic bottom-up}.
In order to formalize it, we need an indicator function which selects  the suitable vectors $\mathbf{u}$:
$$\mathds{1}_{\mathbf{A}\mathbf{b}}(\mathbf{u})= \begin{cases}1 &\text{if }  \mathbf{u} = \mathbf{A}\mathbf{b}\\ 0 & \text{otherwise}\end{cases}.$$ 
The pmf $p_{BU}(\mathbf{y})$ is:
\begin{equation*}
	p_{BU}(\mathbf{y}) = p_{BU} (\mathbf{u},\mathbf{b}) = \ \mathds{1}_{\mathbf{A}\mathbf{b}}(\mathbf{u}) \hat{p}(\mathbf{b}).
\end{equation*}
The pmf $p_{BU}(\mathbf{y})$ assigns positive probability only to vectors $\mathbf{y} \in s(\mathcal{X}^m)$ because of the indicator $\mathds{1}_{\mathbf{A}\mathbf{b}}(\mathbf{u})$; it is thus  coherent.

\subsection{Conditioning on uncertain evidence}\label{sec:soft_evidence}
Virtual evidence \citep{pearl1988probabilistic} is a method for  conditioning a joint distribution 
on an uncertain evidence, obtained for instance  from a noisy source of information.
It is also referred to as soft evidence “nothing else considered”  by \citep[3.9]{darwiche2009modeling} and \citep{chan2005revision}. 

Consider two discrete variables $X$ and $Z$ and their  joint prior  pmf $p(x,z)=p(x)p(z|x)$, where the values of 
$\{z \in z_k\}_{k=1}^{K}$ are mutually exclusive.
Virtual evidence assumes that $\zeta$, i.e. the uncertain observation of $Z$,
can be expressed
by the likelihood ratios: $\lambda_1 : ...: \lambda_K = L(Z=z_1) : ... : L(Z=z_K)$.
Assuming $\zeta$  to be independent from the prior, the update rule  (see \citet[Theorem 2]{chan2005revision} and \citet{munk2022uncertain}) is:
\begin{align}
	p(x | \zeta) = \frac{\sum_{k=1}^{K}  p(x , z_k) \lambda_k}
	{\sum_{j=1z}^{K}  p(z_j) \lambda_j}.
	\label{eq:virtual_ev}
\end{align}

We can make a few observation about the update of Eq.\eqref{eq:virtual_ev}.
First,  $p(x)$ and
$p(x|\zeta)$ have the same zero probabilities.
Indeed, virtual evidence is based on conditioning,   which does not modify the zero probabilities  \cite[Chapter 3.3]{darwiche2009modeling}. 
% In reconciliation, we update  $ p_{BU}(\mathbf{y})$ whose zero probabilities correspond to incoherent vectors. Thus also $\tilde{p}(\mathbf{y})$ will assign zero probabilities to incoherent vectors.

If there is a unique  $\lambda_j$\textgreater 0 and
the  remaining $\lambda_i$ ($i \neq j$)   are zero, we have a certain observation ($Z=z_j$). In this case,  the conditioning of Eq.\eqref{eq:virtual_ev} is equivalent to Bayes' rule.

The evidence $\zeta$ does not need to be normalized, as what matters are the likelihood ratios.
However in our application $\zeta$ is constituted by the base forecast of the upper variables and thus it is
normalized.

\subsection{Reconciliation by conditioning on the base forecast of the upper time series}\label{sec:conditioning}
We now show how to 
use virtual evidence in order to condition
$p_{BU}(\mathbf{u},\mathbf{b})$ on  the base forecasts $\hat{p}(u_1)$ of the first upper  time series $u_1$. %\hl{the first variable u1 o U1??}
Let us denote by $\mathbf{A}_{[1,]}$ the first row of 
$\mathbf{A}$, such that $u_1 = \mathbf{A}_{[1,]} \mathbf{b}$. 
According to Eq.\eqref{eq:virtual_ev},
the reconciled pmf of the bottom time series is:
\begin{equation}
	\tilde{p}_1(\mathbf{b}) =  \dfrac{\sum_{u_1^*} p_{BU}(u_1^*, \mathbf{b}) \hat{p}(u_1^*)}{\sum_{u_1^*} \sum_{\mathbf{b}^*} p_{BU}(u_1^*, \mathbf{b}^*) \hat{p}(u_1^*)},
	\label{eq:Bayesrec1}
\end{equation}
where $p_{BU}(u_1^*, \mathbf{b})$ denotes the marginal of $p_{BU}$ where all other upper time series forecasts are marginalized. The sums in eq.~\eqref{eq:Bayesrec1} are over all possible values $\mathbf{u}_1^*$ and $\mathbf{b}^*$ in the domain of the pmfs $\hat{p}(u_1)$ and $p_{BU}(u_1,\mathbf{b})$. 

The summation $\sum_{u_1^*} p_{BU}(u_1^*, \mathbf{b})$ is sparse, since
$p_{BU}(u_1^*, \mathbf{b})$ is non-zero only if $\mathbf{A}_{[1,]} \mathbf{b} = u_1^*$.
The subscript  in $\tilde{p}_1(\mathbf{b})$ shows that only the base forecast regarding $u_1$  has been considered.\\

Further insights about the update with virtual evidence
can be obtained by analyzing the relative probability of two bottom vectors $\mathbf{b}^*$ and $\mathbf{b}^{**}$, such that
$\mathbf{A}_{[1,]} \mathbf{b}^* = u_1^*$ and
$\mathbf{A}_{[1,]} \mathbf{b}^{**} = u_1^{**}$.
From  Eq.\eqref{eq:virtual_ev} we obtain:
\begin{equation}
	\frac
	{\tilde{p}_1(\mathbf{b}^{**})}{\tilde{p}_1(\mathbf{b}^*)}
	= 
	\frac
	{p_{BU}(u_1^{**}, \mathbf{b})}
	{p_{BU}(u_1^*, \mathbf{b})}
	\cdot 
	\frac
	{\hat{p}(u_1^{**})}
	{\hat{p}(u_1^*)},
	\label{eq:virtual_ev_rel_prob}
\end{equation}
which shows how virtual evidence updates the relative probability
of $\mathbf{b}^{*}$ and $\mathbf{b}^{**}$, merging information from
$p_{BU}(\mathbf{y})$ and $\hat{p}(u_1)$.

The reconciled joint pmf is eventually:\\
\begin{align*}
	\tilde{p}_1(\mathbf{y}) &= \tilde{p}_1(\mathbf{u},\mathbf{b})= \mathds{1}_{\mathbf{A}\mathbf{b}}(\mathbf{u}) \tilde{p}_1(\mathbf{b}) \\
	&= \mathds{1}_{\mathbf{A}\mathbf{b}}(\mathbf{u}) \dfrac{\sum_{u_1^*} p_{BU}(u_1^*, \mathbf{b}) \hat{p}(u_1^*)}{\sum_{u_1^*} \sum_{\mathbf{b}^*} p_{BU}(u_1^*, \mathbf{b}^*) \hat{p}(u_1^*)}.
\end{align*}
Since $\tilde{p}_1(\mathbf{y})$  has been obtained 
by applying virtual evidence on $p_{BU}(\mathbf{y})$, it has the same support of $p_{BU}(\mathbf{y})$, i.e.  $s(\mathbb{N}^m)$.
Thus,  $\tilde{p}_1(\mathbf{y})$ is coherent.

\paragraph{Sequential updates}
The reconciled pmf $\tilde{p}_1(\mathbf{y})$ can be further updated with  new uncertain evidence and the   updates performed via virtual evidence are commutative \citep{chan2005revision,munk2022uncertain}.
Indeed, the method assumes  the conditional independence between the uncertain observations
(the  base forecast of the different upper variables in our application);
this  is a common assumption when merging probabilistic information acquired from different noisy sensors    \cite[Sec. 2]{durrant2016multisensor}.

We thus adopt a sequential approach, performing 
an update of type eq.~\eqref{eq:Bayesrec1}
for the base forecasts of each upper time series.
The first iteration updates $p_{BU}(\mathbf{y})$ with the virtual evidence $\hat{p}(u_1)$ to obtain $\tilde{p}_1(\mathbf{b})$; the second iteration updates $\tilde{p}_1(\mathbf{y})$ with the virtual evidence  $\hat{p}(u_2)$ to obtain $\tilde{p}_2(\mathbf{y})$, and so on. 
Assuming all base forecasts to be available, the final reconciled distribution is $\tilde{p}(\mathbf{y}) := \tilde{p}_{n-m}(\mathbf{y})$. 
If the base forecast of a  certain upper variable is missing, the corresponding update is skipped.

\begin{proposition}
	If the upper time series forecasts are conditionally independent, the sequential updates procedure is equivalent to a full update procedure with 
	\begin{align*}
		\tilde{p}(\mathbf{y}) &= \tilde{p}(\mathbf{u},\mathbf{b}) = \mathds{1}_{\mathbf{A}\mathbf{b}}(\mathbf{u}) \tilde{p}_1(\mathbf{b}) \\
		&= \mathds{1}_{\mathbf{A}\mathbf{b}}(\mathbf{u}) \dfrac{\sum_{\mathbf{u}^*} p_{BU}(\mathbf{u}^*, \mathbf{b}) \hat{p}(\mathbf{u}^*)}{\sum_{\mathbf{u}^*} \sum_{\mathbf{b}^*} p_{BU}(\mathbf{u}^*, \mathbf{b}) \hat{p}(\mathbf{u}^*)}.
	\end{align*}
\end{proposition}

\subsection{Reconciling a  Minimal Hierarchy}
\label{sec:toyExample}

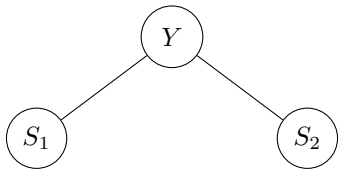
\begin{figure}[htp!]
	\centering
	\begin{tikzpicture}[level/.style={sibling distance=40mm/#1}, , scale=0.9]
		\node [circle,draw] (z){$\: Y\:$ }
		child {node [circle,draw]  {$S_1$}
		}
		child {node [circle,draw] {$S_2$}
		};
	\end{tikzpicture}
	\caption{The minimal hierarchy.}
	\label{fig:toy-hierarchy}
\end{figure}

\begin{table}[ht!]
	\begin{subtable}[h]{\textwidth}
		\centering
		\begin{tabular}{@{}cccc@{}}
			\toprule
			\footnotesize $S_1=0$, $S_2=0$&
			\footnotesize $S_1=0$, $S_2=1$&
			\footnotesize $S_1=1$, $S_2=0$&
			\footnotesize $S_1=1$, $S_2=1$\\
			\midrule
			.25& .25 & .25 & .25 \\
			\bottomrule
		\end{tabular}
		\caption{$\hat{p}(s_1,s_2)$.}\label{tab:jointS1S2}
	\end{subtable}
	
	%\newline
	\vspace*{0.5 cm}
	%\newline
	
	\begin{subtable}[h]{\textwidth}
		\centering
		\begin{tabular}{@{}ccccc@{}}
			\toprule
			&
			\footnotesize $S_1=0$, $S_2=0$&
			\footnotesize $S_1=0$, $S_2=1$&
			\footnotesize $S_1=1$, $S_2=0$& 
			\footnotesize $S_1=1$, $S_2=1$\\
			\midrule
			$Y=0$&$1$&$0$&$0$&$0$\\
			$Y=1$&$0$&$1$&$1$&$0$\\
			$Y=2$&$0$&$0$&$0$&$1$\\
			\bottomrule
		\end{tabular}
		\caption{$\mathds{1}_{s_1+s_2}(y)$.}\label{tab:delta}
	\end{subtable}
	
	%\newline
	\vspace*{0.5 cm}
	%\newline
	
	\begin{subtable}[h]{\textwidth}
		\centering
		\begin{tabular}{@{}ccccc@{}}
			\toprule
			&
			\footnotesize $S_1=0$, $S_2=0$&
			\footnotesize $S_1=0$, $S_2=1$&
			\footnotesize $S_1=1$, $S_2=0$&
			\footnotesize $S_1=1$, $S_2=1$\\
			\midrule
			\rowcolor{backcolour} $Y=0$&$.25$&0&0&0\\
			\rowcolor{backcolour} $Y=1$&0&$.25$&$.25$&0\\
			\rowcolor{backcolour} $Y=2$&0&0&0&$.25$ \\ 
			\bottomrule
		\end{tabular}
		\caption{$p_{BU}(y,s_1,s_2) = \mathds{1}_{s_1+s_2}(y) \cdot \hat{p}(s_1,s_2)$.}\label{tab:jointBU}
	\end{subtable}
	\caption{Probabilistic bottom-up reconciliation.}
\end{table}

Figure \ref{fig:toy-hierarchy}  represents a minimal temporal hierarchy, whose bottom variables are the  two semesters  and whose upper variable is the year.
We assume  $S_1$ and $S_2$ to  take values in $\{0, 1\}$, $Y$ to take values in $\{0,1,2\}$, the data  to arrive up to year $t$ and the base forecast to refer to year $t+1$.
We denote by $S_{(t,i)}$ the random variable corresponding to the value of the $i$-th semester  of year $t$ and by $s_{(k,i)}$ the observation referring to 
the $i$-th semester  of year $k$ ($k$ \textless $t$).
Moreover, $Y_{t+1}$ denotes the random variable corresponding to year $t+1$, while $y_k$ denotes the observation of year $k$. 
The probability mass functions of the base forecasts are thus:
\begin{align*}
	\hat{p}(s_{1}) & := \hat{p}(s_{t+1,1})  =  p(s_{t+1,1} \mid s_{(1,1)},..,s_{(t,2)} ),  \\ \nonumber
	\hat{p}(s_{2}) & := \hat{p}(s_{t+1,2})  = p(s_{t+1,2} \mid s_{(1,1)},..,s_{(t,2)} ),  \\\nonumber
	\hat{p}(y) & := \hat{p}(y_{t+1})  = p(y_{t+1} \mid y_1,..,y_t ),
\end{align*}
where we introduce a simplified notation which drops the time from the subscript. 

We obtain the joint distribution of the bottom variables assuming 
$\hat{p}(s_1,s_2) = \hat{p}(s_1) \hat{p}(s_2)$. 
In this paper, we always use this independence assumption to create
the joint mass function of the bottom variables.
However, this is not a requirement of our method, which
could also reconcile a predictive multivariate   distribution. 
We leave this as a future research work, acknowledging that 
modelling correlations in temporal hierarchies \citep{nystrup2020temporal} is  an important problem.

Eventually, the pmf of probabilistic bottom-up reconciliation is:
\begin{equation*}\label{eq:joint2}
	p_{BU}(y,s_1,s_2) 
	= \mathds{1}_{s_1+s_2}(y)  \hat{p}(s_1,s_2) \,,
\end{equation*}
We provide a numerical example  in Table~\ref{tab:jointBU}, assuming 
$\hat{p}(s_1, s_2)$ to be uniform.

\subsubsection*{Conditioning on $\hat{p}(y)$}

\begin{table}[!h]
	\begin{subtable}[h]{\textwidth}
		\centering
		\begin{tabular}{@{}cc@{}}
			& \\
			\toprule
			$\text{Y}$ & $\hat{p}(\text{y})$\\
			\midrule
			$0$ & .5\\
			$1$ & .2\\
			$2$ & .3 \\ 
			\bottomrule
		\end{tabular}
		\caption{$\hat{p}(y)$.}\label{tab:py_bu}
	\end{subtable}
	
	\vspace*{0.5 cm}
	\begin{subtable}[h]{\textwidth}
		\centering
		\begin{tabular}{@{}ccccc@{}}
			\toprule
			&
			\footnotesize $S_1=0$, $S_2=0$&
			\footnotesize $S_1=0$, $S_2=1$&
			\footnotesize $S_1=1$, $S_2=0$&
			\footnotesize $S_1=1$, $S_2=1$\\
			\midrule
			&
			$ \frac{.25 \cdot .5}{c} = .416$ &
			$ \frac{.25 \cdot .2}{c} = .167$ &
			$ \frac{.25 \cdot .2}{c} = .167$ &
			$\frac{.25 \cdot .3}{c} = .25$ \\
			\bottomrule
		\end{tabular}
		\caption{$\tilde{p}(s_1,s_2)$.
			The normalizing constant is $c = .25 \cdot .5 +
			25 \cdot .2 + .25 \cdot .2 + .25 \cdot .3$.
		}\label{tab:joint}
	\end{subtable}
	
	\vspace*{0.5 cm}
	\begin{subtable}[h]{\textwidth}
		\centering
		\begin{tabular}{@{}ccccc@{}}
			\toprule
			\footnotesize $\tilde{p}(y,s_1,s_2)$&
			\footnotesize $S_1=0$, $S_2=0$&
			\footnotesize $S_1=0$, $S_2=1$&
			\footnotesize $S_1=1$, $S_2=0$&
			\footnotesize $S_1=1$, $S_2=1$\\
			\midrule
			\rowcolor{backcolour} 
			$Y=0$ & $.416$ & 0 & 0 & 0\\
			\rowcolor{backcolour} 
			$Y=1$ & 0 & $.167$ & $.167$ & 0\\
			\rowcolor{backcolour} 
			$Y=2$ & 0 & 0 & 0 & $.25$ \\ 
			\bottomrule
		\end{tabular}
		\caption{$\tilde{p}(y,s_1,s_2) = \mathds{1}_{s_1+s_2}(y) \tilde{p}(s_1,s_2) $}\label{tab:joint-updated}
	\end{subtable}
	\caption{Reconciliation of the minimal hierarchy using  virtual evidence.}\label{tab:recon-prob}
\end{table}

We now update $p_{BU}(y,s_1,s_2)$  by conditioning on $\hat{p}(y)$.
By applying  the  updating of Eq.\eqref{eq:Bayesrec1}, we have:

\begin{align*}\label{eq:softevi}
	\tilde{p}(s_1,s_2) & =
	\frac{\sum_{y\in\{0,1,2\}} 
		\hat{p}(y) \cdot p_{BU}(y,s_1,s_2)}
	{\sum_{y'\in\{0,1,2\}}\sum_{s_1',s_2'\in\{0,1\}} \hat{p}(y') \cdot p_{BU}(y',s_1',s_2')}\, 
\end{align*}

and hence:
\begin{align*}%\label{eq:softevi2}
	\tilde{p}(s_1,s_2, y) &= \mathds{1}_{s_1+s_2}(y) \cdot \tilde{p}(s_1,s_2).	
\end{align*}

In  Table \ref{tab:recon-prob} we show a numerical example.

\subsection{Reconciling Poisson base forecast}
\label{sec:poisson}

We now consider an example with Poisson base forecast.
We denote by  $\text{Poi}(x | \lambda_X)$  the Poisson pmf  with parameter $\lambda_X$, and 
we  assume the base forecasts to be: 

\begin{align*}
	%\label{eq:PoissonExample}
	\hat{p}(s_1) &  = \text{Poi}(s_1|\lambda_{1})\,, \\ \nonumber
	\hat{p}(s_2) & = \text{Poi}(s_2|\lambda_{2})\,, \\ \nonumber
	\hat{p}(y) & = \text{Poi}(s_y|\lambda_Y)\,,
\end{align*}

The bottom-up pmf is:
\begin{equation*}
	p_{BU}(s_{1}, s_{2}, y)  = \text{Poi}(s_{1} | \lambda_1) \text{Poi}(s_{2} | \lambda_2)  \mathds{1}_{s_1+s_2}(y).
\end{equation*}

while the reconciled pmf  of the bottom variables  is:
\begin{equation}
	\label{eq:poi-reconc}
	\tilde{p}(s_1, s_2)  =
	\frac{\sum_{y=0}^{+\infty}  
		\overbrace{\text{Poi}(s_{1} | \lambda_1) \text{Poi}(s_{2} | \lambda_2) \mathds{1}_{s_1+s_2}(y) }^{p_{BU}(y,s_1,s_2)}
		\overbrace{\text{Poi}(y | \lambda_Y)}^{\hat{p}(y)}
	} {\sum_{y=0}^{+\infty} \sum_{s_1}\sum_{s_2} \text{Poi}(s_{1} | \lambda_1) 
		\text{Poi}(s_{2} | \lambda_2) \mathds{1}_{s_1+s_2}(y) \text{Poi}(y | \lambda_Y)},
\end{equation}
which is analytically intractable.

\subsection{Sampling the reconciled distribution}
\label{subsec:probProgr}
The  reconciled pmf of Eq.~\eqref{eq:poi-reconc} can be however computed
via sampling.
This is for instance an implementation based on PyMC3  \citep{pymc3, martin2018bayesian}, a package for automatic Bayesian inference:

\begin{lstlisting}[language=Python, mathescape=true]
	def reconcile (lambda1, lambda2, lambdaY):
	import pymc3 as pm
	basic_model = pm.Model()
	with basic_model:
	#base forecast of S1 and S2
	S1  = pm.Poisson  ('S1', mu = lambda1)
	S2  = pm.Poisson  ('S2', mu = lambda2)
	
	#virtual evidence
	Y   = pm.Poisson  ('Y',  mu = lambdaY, observed = S1 + S2)
	#implies updating p(s1,s2) with p(s1,s2) * p(Y = s1+s2)
	
	#sampling the reconciled pmf
	trace = pm.sample()
	return trace
	
\end{lstlisting}
The probabilistic program returns samples  from $\tilde{p}(s_1,s_2)$,
using the Metropolis-Hasting algorithm.
In general the probabilistic program  contains $m$   base forecasts  and ($n-m$)   virtual evidences (indicated by the keyword "observed"), one for each upper variable of the hierarchy.
As an example, we provide in~\ref{sec:421hierarchy} the code which reconciles a 4-2-1 hierarchy. 
Alternative packages for probabilistic programming such as Stan \citep{stan} could  be used in a  similar way. 

\subsection{A numerical example}
We now report the results  assuming $\lambda_1 = 2, \lambda_2 = 4, \lambda_Y =9$.
Given the positive incoherence ($\lambda_Y > \lambda_1 + \lambda_2$), 
reconciliation increases 
the expected value of both $S_1$ and $S_2$: see Table~\ref{tab:reconcToy}. 
and the left plot of Fig.~\ref{fig:reconc-B1}. 
A larger increase is applied to  the variable whose base forecast has larger variance, i.e.,  
$S_2$ (Table~\ref{tab:reconcToy}). Moreover, the variances of $S_1$, $S_2$ and $Y$ decrease after reconciliation, since novel information has been acquired through conditioning. 
These are the same  patterns already reported  for the probabilistic Gaussian reconciliation \citep{corani_reconc}.

\renewcommand\arraystretch{1.2}
\begin{table}[htp!]
	\centering
	\begin{tabular}{lccb||ccb}
		\toprule
		& \multicolumn{3}{c}{\cellcolor{lightyellow}\textbf{mean}} & \multicolumn{3}{c}{\cellcolor{lightyellow} \textbf{var}} \\
		& $p_{BU}$ & $\tilde{p}$ & $\Delta$ & $p_{BU}$ & $\tilde{p}$ &  $\Delta$  \\ \midrule
		$S_1$ &  2.0  & 2.4  & +0.4 & 2.0 & 1.9 & -0.1 \\
		$S_2$ &  4.0  & 4.8  & +0.8 & 4.0 & 3.0 & -1.0 \\ 
		$Y$   &  9.0  & 7.2  & -1.8 & 6.0 & 3.6 & -2.4 \\ 	\bottomrule
	\end{tabular}
	\caption{\label{tab:reconcToy} 
		Reconciliation results
		for the example $\lambda_1$=2, $\lambda_2$=4, $\lambda_Y$=9.}
\end{table}    

\begin{figure}[ht!]
	\centering
	\includegraphics[width=0.48\linewidth]{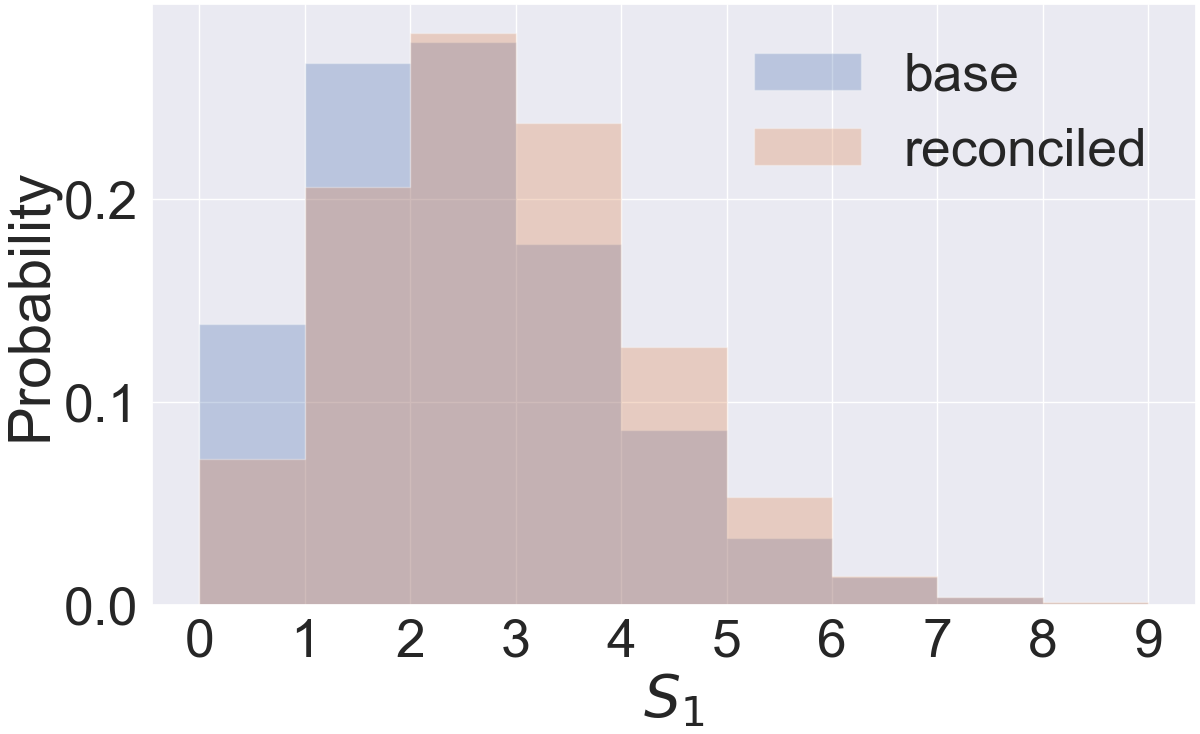}
	\hspace{.2cm}
	\includegraphics[width=0.48\linewidth]{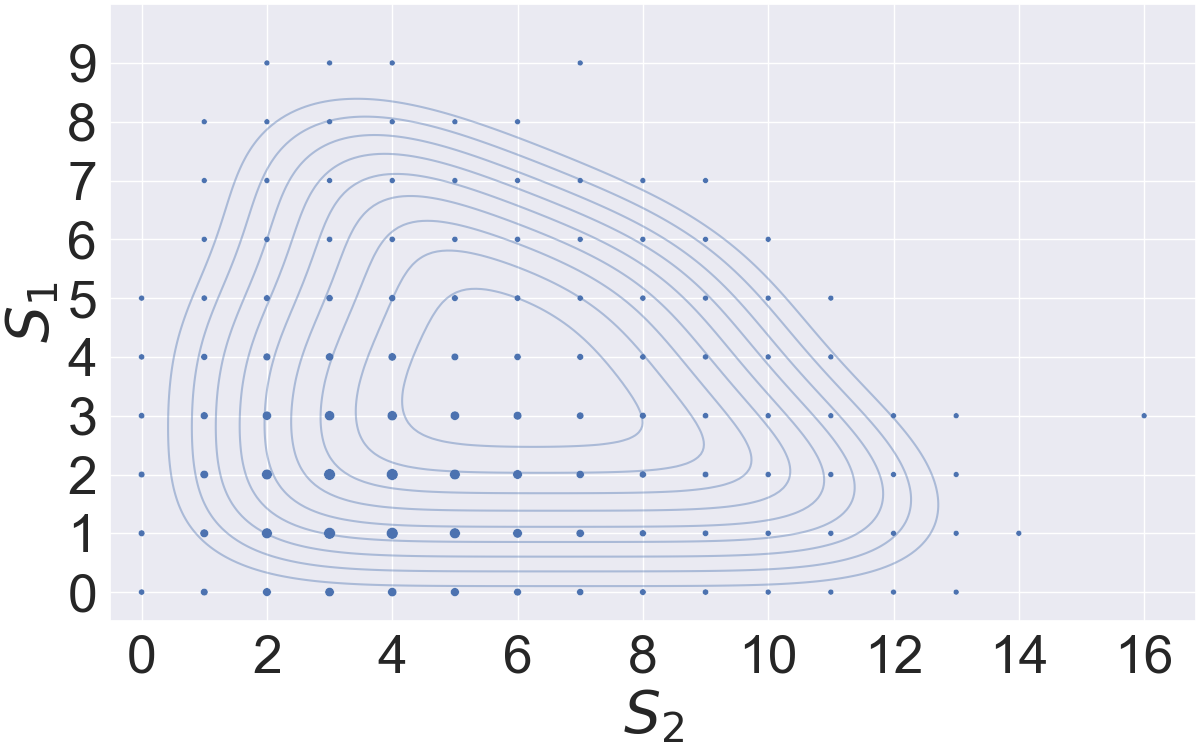}
	\caption{\textit{Left}: The mean  of  $S_1$ increase after reconciliation. \textit{Right}:  $S_1$ and $S_2$ are negatively correlated
		after reconciliation.}
	\label{fig:reconc-B1}
\end{figure}

The right plot of Figure~\ref{fig:reconc-B1} shows that 
$S_1$ and $S_2$ become  negatively correlated after reconciliation.
Indeed, $S_1$ and $S_2$ become dependent once $Y$ is observed, because  $S_1 + S_2 = Y$.
If for instance we  observe $Y$=1, 
the only joint states compatible with the evidence 
are ($S_1=0, S_2=1$) and ($S_1=1, S_2=0$), whence  the negative correlation. 
For the same reason, also virtual evidence induces   negative correlation.
Fig.~\ref{fig:reconc-U}  shows that the reconciled pmf of $Y$ is a compromise between 
its bottom-up pmf and its base forecast:
this  is an analogy with  Bayesian inference, where the posterior distribution  is a compromise between the prior  and the likelihood of the observation.

\begin{figure}[ht!]
	\centering
	\includegraphics[width=0.58\linewidth]{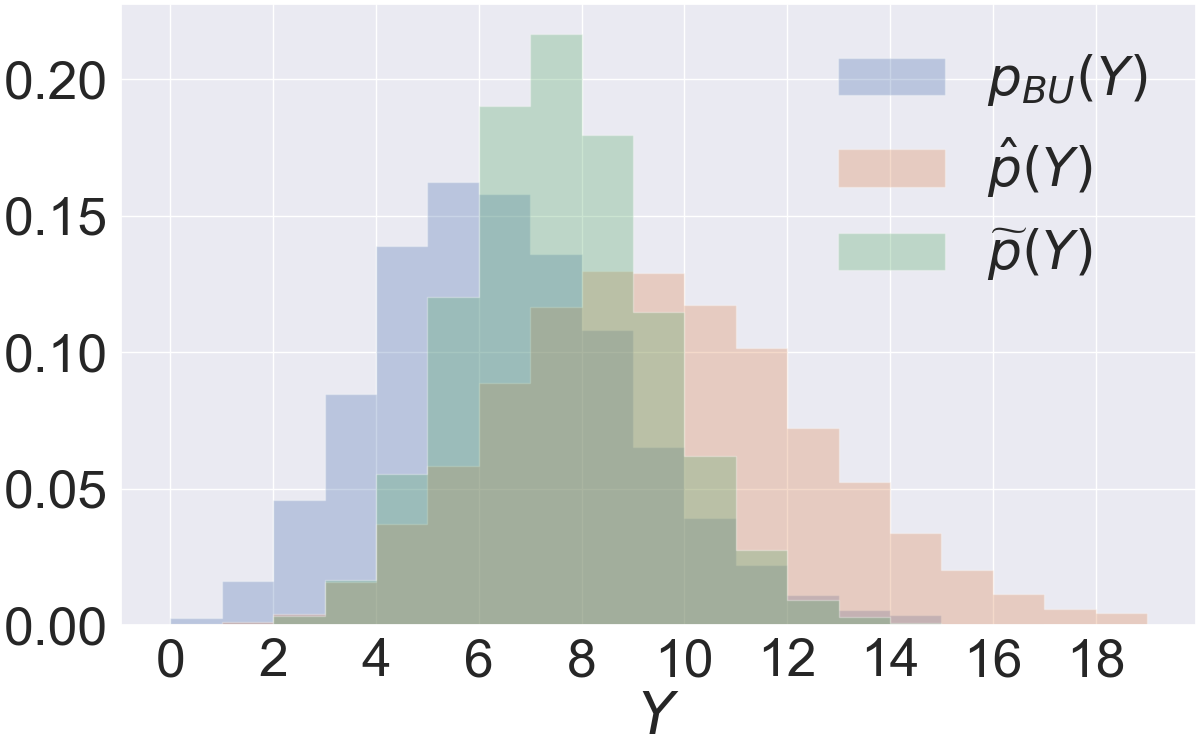}
	\caption{Bottom-up, base forecast  and reconciled pmf  of $Y$. 
	}
	\label{fig:reconc-U}
\end{figure}

\section{Experiments}
\label{sec:experiments}
We select time series 
with maximum value \textless 30 and with average inter-demand interval (ADI) \textless 2, which can be appropriately modelled by  autoregressive models for counts.
Following this criterion we select:
\begin{itemize}
	\item 219 time series   from   \textit{carparts}, available from the R package \textit{expsmooth} \citep{expsmooth} and regarding
	monthly sales of car parts;
	\item 53 time series  from  \textit{syph}, available from the R package \emph{ZIM} \citep{zim_manual} and
	regarding the weekly number of syphilis cases in the United States,  which we aggregate to the monthly scale (this step involve some   approximation);
	\item 135 time series from   \textit{hospital},
	available from the R package \textit{expsmooth} \citep{expsmooth} and
	regarding the	monthly  counts of patients. 
\end{itemize}
In Table~\ref{tab:dsets} we report the percentage of intermittent time series in each data set.

\begin{table}[htp!]
	\centering
	\begin{tabular}{@{}rrrr@{}}
		\toprule
		& selected & \% intermittent & mean length\\ 
		&  &  & \multicolumn{1}{c}{(years)}\\ 
		
		\midrule
		\rowcolor{backcolour}
		carparts &                      219 &                            94\% &   4                         \\
		syph &                      53 &                       47\% &        4                    \\
		\rowcolor{backcolour}
		hospital &                     135 &                             0\% &      7                      \\ \bottomrule
	\end{tabular}
	\caption{Main characteristics of the selected time series. We consider a time series as \textit{intermittent} if its average inter-demand interval (ADI) is \textgreater 1.32 \citep{SyntetosBoylan2005}. 
	}\label{tab:dsets}
\end{table}

\paragraph{Temporal hierarchy and base forecasts}
In every experiment  the bottom forecasts are at the monthly scale.
As in \cite{athanasopoulos2017_temporal},
we compute  the following  temporal aggregates: 2-months, 3-months, 4-months, 6-months, 1 year.

At each level of the hierarchy we  fit a GLM autoregressive model 
with  negative binomial predictive distribution.
We use the  \textit{tscount}  package \citep{tscount} and we select via BIC the order of the autoregression. 

At every level of the hierarchy, the test set  has length of one year.
We thus compute  forecasts up to $h$=12 steps ahead at the monthly level,
up to $h$=6 steps ahead 
at the bi-monthly level, etc.
The forecast of count time series models 
have closed form only one step ahead;  after that,
the  predictive distribution is constituted by samples \citep{tscount}.
Depending on the reconciliation method being adopted, we fit a Gaussian or a negative binomial distribution on the  samples.

\paragraph{Reconciliation}
In the Gaussian case \citep{corani_reconc}, the   reconciled predictive distribution has the same mean and variance of minT  (formulas~\eqref{eq:s_ytilde}
and~\eqref{eq:varianceMinT} respectively)
For  Gaussian reconciliation,
we fit a Gaussian distribution on the samples of the base forecast, for each level and for each   forecasting horizon (e.g., we fit 12 different distributions  at the monthly level).
We then perform  reconciliation using the covariance matrix of 
hierarchy variance and the structural scaling, discussed in Sec \ref{sec:t-hier}.
These methods are  referred to in the following as  \textit{normal} and \textit{structural scaling}.

We moreover implemented a \textit{truncated} approach.
To this end we perform
the \textit{normal} reconciliation,
truncating the distribution of the reconciled bottom forecast. 
We then sum them up via sampling in order to obtain the distribution of the  upper variables.
This a simple way to  obtain positive reconciled forecasts. 

We implement our approach by fitting a  negative binomial distribution on the  base forecasts on the samples of  each level and  each  forecasting horizon and performing reconciliation  via probabilistic programming. The probabilistic program  contains 12 variables (the bottom variables) and 16 soft evidences, corresponding to the upper variables of the hierarchy. 
We refer to this methods as \textit{probCount}.
The reconciliation takes about 2-3  minutes 
on a standard laptop.  This  approach is therefore currently not suitable to  hierarchies containing large number of variables.
Alternative approaches based on importance sampling   
constitute a promising direction for efficient probabilistic reconciliation  
\citep{zambon2022}.

\paragraph{Indicators}
We assess the methods according to their point forecasts, predictive distributions and prediction intervals.
Let us denote by $y_{t+h}$ the actual value of the time series at time $t+h$
and by $\hat{y}_{t+h|t}$ the point forecast computed at time $t$ for time $t+h$. We further denote the predictive distribution by $\hat{f}_{t+h|t}$. Note that $\hat{f}_{t+h|t}(i)$ is a discrete probability mass for $i=1, \ldots, \infty$. 

The mean scaled absolute error (MASE) \citep{hyndman2006another} is defined as:
\begin{align*}
	\text{MAE} & = \frac{1}{h} \sum_{j=1}^{h}|y_{t+j}-\hat{y}_{t+j|t}|\,,\\
	Q & = \frac{1}{T-1}\sum_{t=2}^{T}|y_t - y_{t-1}|\,,\\
	\text{MASE} & = \frac{MAE}{Q}\,.\\
\end{align*}

We use the median of the reconciled distribution  as  point forecast, since 
it is the optimal point forecasts for MAE and MASE \citep{kolassa2016evaluating}.
However,   the median point forecast is generally not coherent, even if the joint distribution  is coherent \citep{kolassa2022}.

Different scores for discrete predictive distributions are discussed in \citep{kolassa2016evaluating}. Here we use the ranked probability score (RPS). Given the predictive probability mass  $\hat{f}_{t+h|t}(i)$, the cumulative predictive probability mass is $\hat{F}_{t+h|t}(k)=\sum_{i=0}^{k}\hat{f}_{t+h|t}(i)$. For a realization $y_{t+h}$, then we have:
\begin{equation}
	\text{RPS}(\hat{f}_{t+h|t}, y_{t+h}) = \sum_{k=0}^\infty (\hat{F}_{t+h|t}(k) - \mathbb{1}_{y_{t+h} \leq k})^2,
\end{equation}
where $\mathbb{1}_{y \leq k}$ is the indicator function for $y\leq k$. 

We compute the RPS of continuous distributions by applying the continuity correction, i.e. computing $p(X=x)$ as $\int_{x-0.5}^{x+0.5}g(t)dt$, where $g(t)$ denotes the continuous density.

We score the prediction intervals via
the mean interval score (MIS)
\citep{gneiting2011quantiles}.
Let us denote by (1-$\alpha$) 
the desired coverage of the interval,
by $l$ and $u$ the lower and upper bounds  
of the interval. 
We have: 
\begin{equation*}
	\text{MIS}(l,u,y) =
	(u-l) + \frac{2}{\alpha} (l-y) \mathbb{1}(y < l) + \frac{2}{\alpha} (y-u) \mathbb{1}(y > u)\,.
\end{equation*}
We adopt a 90\% coverage level ($\alpha$=0.1).
The  MIS rewards narrow prediction intervals; however, it also penalizes intervals which do not contain the actual value; the  penalty  depends on $\alpha$.
In the definition of RPS and MIS, it is understood that $y_{t+j}$, $\hat{f}_{t+h|t}(i)$, $\hat{F}_{t+h|t}(k)$, $l$ and $u$ are specific for a certain level of the hierarchy and for a certain forecasting horizon $j$, $1 \leq j \leq h$.

We also report the Energy score (ES), which is a proper scoring rule for distributions defined on the entire hierarchy
\citep{panagiotelis2022}.
Given a realization $\mathbf{y}$ and a joint probability $P$ on the entire hierarchy, the ES is: 
\begin{equation}
	ES(P,\mathbf{y}) = E_P 
	|| \mathbf{y} - \mathbf{s}||^{\alpha} 
	- \frac{1}{2} E_P
	|| \mathbf{s} - \mathbf{s}^*||^{\alpha} \,,
\end{equation}
where $\mathbf{s}$ and $\mathbf{s}^*$ are independent samples drawn from $P$.
We compute the energy score using the \textit{scoringRules} package\footnote{We use R packages  in Python via the \textit{rpy2} utility,
	\url{https://rpy2.github.io}.} \citep{pkg:scoringRules} with $\alpha$=2.

\paragraph{Skill score}
We compute the
\textit{skill score} 
on a certain indicator 
as the percentage improvement  with respect to 
the \textit{normal} method, taken as
a baseline. 
Skill scores  are scale-independent and 
can be thus  averaged across multiple time series.   For instance the skill score of \textit{probCount} on   MASE is:
\[
\text{Skill (MASE, probCount)} = 
\frac{\text{MASE(normal) - MASE(probCount)}}
{(\text{MASE(normal) + MASE(probCount)})/2}\,.
\]

Thus a positive skill score  implies an improvement  compared to \textit{normal}.
The denominator makes the skill score symmetric
and bounded between $-2$ and $2$, allowing a fair comparison between the competitors and  the baseline.

For each level of the hierarchy, we compute the skill score  for each forecasting horizon $j$ ($1 \leq j \leq h$); then we average over the different  $j$.
Only for the \textit{probCount} we compute also the skill score with respect to the base forecast, constituted by a negative binomial distribution  fitted on the samples returned by the GLM models.

\subsection{Experiments on \textit{carparts} and \textit{syph}}

\renewcommand\arraystretch{1.2}
\begin{table}[htp!]
	\begin{subtable}[h]{\linewidth}
		\begin{tabular}{rrrrrb}
			\toprule
			\textbf{Skill score on carparts} &  & \multicolumn{3}{c}{vs \textit{normal}} &  vs \textit{base} \\
			& & \textit{struc scal} & \textit{truncated} & \textit{probCount} & \textit{probCount}  \\
			\midrule
			\rowcolor{lightyellow}
			\textbf{ENERGY SCORE} & & -0.06 &     -0.19 &   0.27 &   0.34 \\ \midrule
			\textbf{MASE} &          Monthly &      -0.01 &     -0.02 &   0.18 &   0.00 \\
			&        2-Monthly &      -0.02 &     -0.08 &   0.21 &   0.00 \\
			&        Quarterly &      -0.03 &     -0.11 &   0.21 &   0.00 \\
			&        4-Monthly &      -0.03 &     -0.14 &   0.30 &   0.18 \\
			&         Biannual &      -0.04 &     -0.20 &   0.30 &   0.09 \\
			&           Annual &      -0.09 &     -0.30 &   0.89 &   0.80 \\
			\rowcolor{lightyellow} & \textit{average} &      -0.04 &     -0.14 &   0.35 &   0.18 \\
			\midrule
			\textbf{MIS} &          Monthly &       0.00 &      0.27 &   0.36 &   0.38 \\
			&        2-Monthly &       0.00 &     -0.07 &   0.15 &   0.53 \\
			&        Quarterly &      -0.01 &     -0.21 &   0.15 &   0.58 \\
			&        4-Monthly &      -0.10 &     -0.29 &   0.17 &   0.67 \\
			&         Biannual &      -0.11 &     -0.27 &   0.20 &   0.85 \\
			&           Annual &      -0.24 &     -0.56 &   0.25 &   1.22 \\
			\rowcolor{lightyellow} & \textit{average} &      -0.08 &     -0.19 &   0.22 &   0.71 \\
			\midrule
			\textbf{RPS} &          Monthly &      -0.02 &     -0.06 &   0.43 &   0.20 \\
			&        2-Monthly &      -0.03 &     -0.12 &   0.37 &   0.29 \\
			&        Quarterly &      -0.04 &     -0.15 &   0.40 &   0.25 \\
			&        4-Monthly &      -0.05 &     -0.20 &   0.42 &   0.33 \\
			&         Biannual &      -0.08 &     -0.27 &   0.32 &   0.43 \\
			&           Annual &      -0.12 &     -0.36 &   1.14 &   0.96 \\
			\rowcolor{lightyellow} & \textit{average} &      -0.06 &     -0.19 &   0.51 &   0.41 \\ 
			\bottomrule
		\end{tabular}
		\caption{Results on time series extracted from \textit{carparts}, detailed by each level of the hierarchy.}
		\label{tab:carparts}
	\end{subtable}
	\newline
	\vspace*{1 cm}
	\newline
	\begin{subtable}[h]{\linewidth}
		\centering
		\begin{tabular}{llrrrb}
			\toprule
			\textbf{Skill score on syph} &  & \multicolumn{3}{c}{vs \textit{normal}} &  vs \textit{base} \\
			& & \textit{struc scal} & \textit{truncated} & \textit{probCount} & \textit{probCount}  \\
			\midrule
			ENERGY SCORE  & &    -0.07& -0.21 &  0.27  &   0.28 \\
			MASE  & &    -0.06& -0.17 &  0.35  &   0.06 \\
			MIS   & &    -0.13& -0.16 &  0.20  &   0.23 \\
			RPS  & &    -0.06 & -0.18 &  0.42 &   0.41 \\
			\bottomrule
		\end{tabular}
		\caption{Results on time series extracted from \textit{syph},  averaged over the entire hierarchy.}
		\label{tab:syph}
	\end{subtable}
	\caption{Skill score on \textit{carparts} and \textit{syph}. The first columns report skill score with respect to \textit{normal}. The last column reports the skill score of \textit{probCount} with respect to the \textit{base} forecasts.}
\end{table}
\begin{figure}[htp!]
	\centering
	\begin{subfigure}[h]{\figwidth}
		\includegraphics[width=\textwidth]{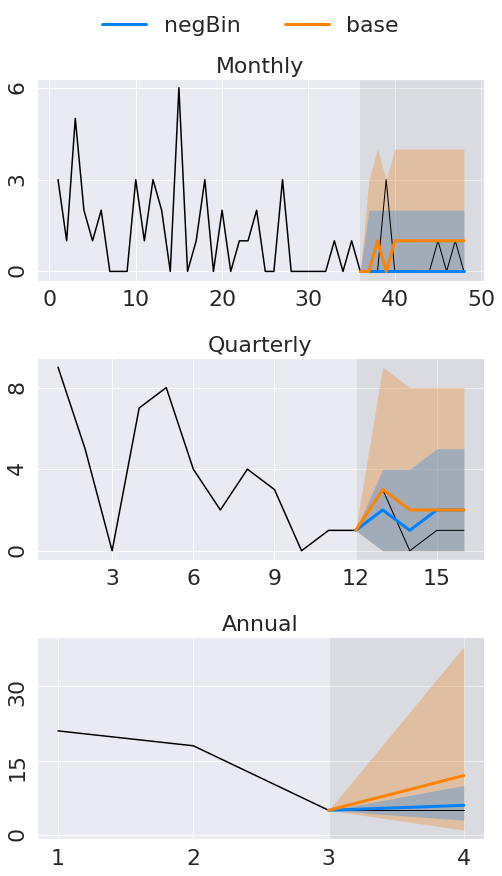}
		\caption{Comparison of forecasts (\textit{base} vs \textit{probCount}) on a time series extracted from \textbf{carparts}.}
		\label{fig:carparts-base}
	\end{subfigure}
	\begin{subfigure}[h]{\figwidth}
		\includegraphics[width=\textwidth]{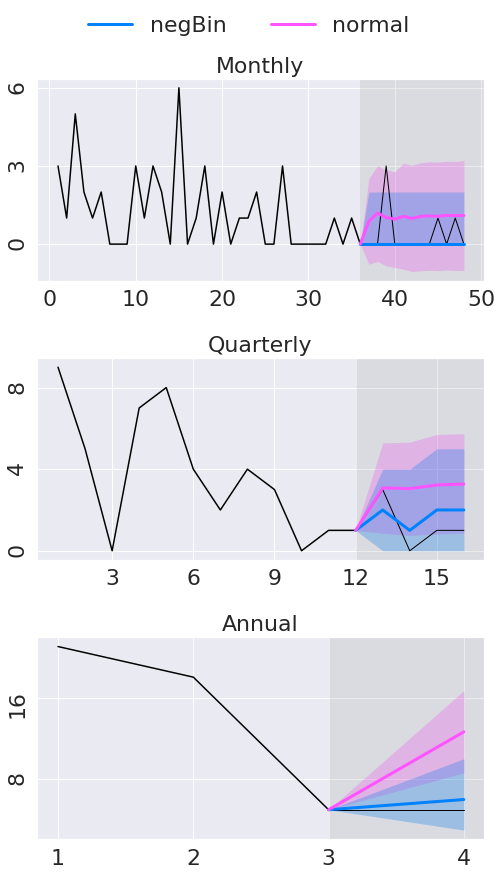}
		\caption{Comparison of forecasts (\textit{normal} vs \textit{probCount}) on a time series extracted from \textbf{carparts}.}
		\label{fig:carparts-normal}
	\end{subfigure}
	
	%\newline
	\vspace*{0.5 cm}
	%\newline
	
	\begin{subfigure}[h]{\figwidth}
		\includegraphics[width=\textwidth]{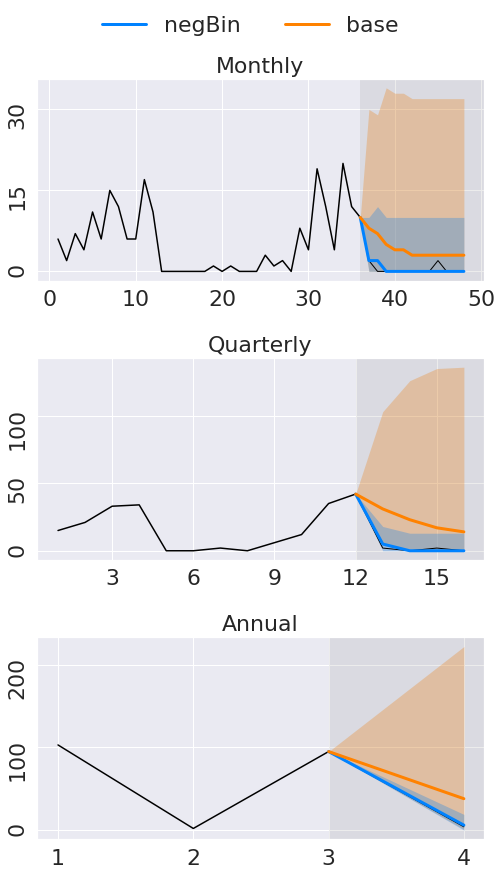}
		\caption{Comparison of forecasts (\textit{base} vs \textit{probCount}) on a time series extracted from \textbf{syph}.}
		\label{fig:syph-base}
	\end{subfigure}
	\begin{subfigure}[h]{\figwidth}
		\includegraphics[width=\textwidth]{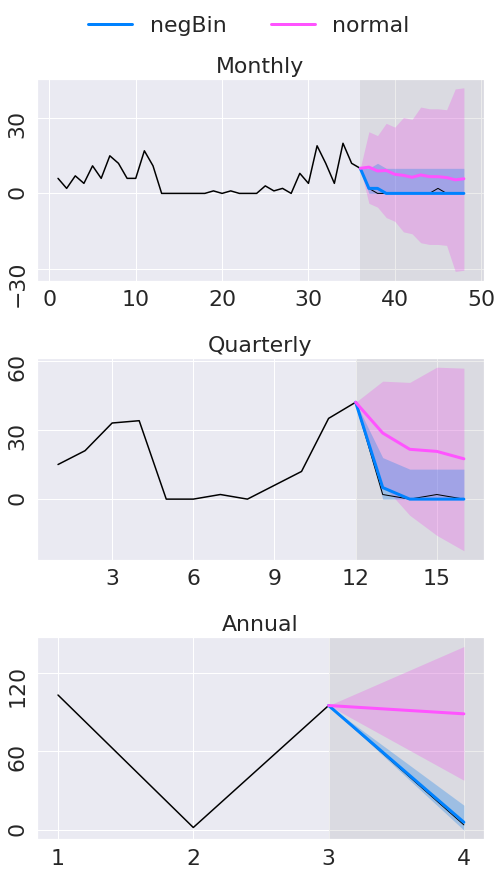}
		\caption{Comparison of forecasts (\textit{normal} vs \textit{probCount}) on a time series extracted from \textbf{syph}.}
		\label{fig:syph-normal}
	\end{subfigure}
	\caption{Reconciliation of two time series from the \textit{carparts} and the \textit{syph} data set. For simplicity we only show three levels of the hierarchy. The black line shows the actual values.}
	\label{fig:carparts-syph}
\end{figure}

The  \textit{carparts} data set has a high percentage of intermittent time series (94\%) and 
the base forecasts are generally asymmetric.  In these conditions 
\textit{probCount} yields a large improvement 
over \textit{normal}
on every score (Table~\ref{tab:carparts}).
Averaging over the entire hierarchy, the  improvement of  \textit{probCount} over \textit{normal} ranges between 22\% and 51\% 
depending on the indicator; the energy score improves by about 27\% 
The \textit{structural scaling} and the  \textit{truncated} method
perform worse than  the \textit{normal} method.
Hence, the \textit{truncated} method
does not   represent a satisfactory solution for modelling distributions over counts,
even if it yields positive forecasts.

The forecast reconciled by \textit{probCount} 
yield a large improvement also compared
to the  base forecast (last column of Table~\ref{tab:carparts}). The largest improvements  with respect to the base forecasts are in the highest level of the hierarchy, as already observed for  temporal hierarchies \citep{athanasopoulos2017_temporal}.

Also in the \textit{syph} data set, \textit{probCount} largely outperforms  \textit{normal} on every indicator (Table~\ref{tab:syph}); the improvement varies between 27\% and 42\%. 
Large improvements are found also with respect to the base forecasts.
The performance of both
\textit{truncated} 
and \textit{structural scaling} 
is  slightly worse than 
\textit{normal} also in this case.
The result for \textit{syph}, detailed for each level of the hierarchy, are given in the appendix.

In Figure~\ref{fig:carparts-syph} 
we provide two examples of reconciliation, taken from \textit{carparts} and \textit{syph} respectively. 
In both examples, the distribution of the base forecasts is asymmetric at every level  (Figures~\ref{fig:carparts-base},~\ref{fig:carparts-normal}), with the median much   lower than the mid-point of the prediction interval. 
Based on this information, \textit{probCount}  revises downwards
the point forecasts compare to the base forecasts.
At the monthly level its point forecasts (i.e., the medians) are  0.
This is both the lower bound of the prediction interval and the median:
the reconciled distribution is strongly asymmetric
as it can happen when the counts are low.
The adjustment applied by
\textit{probCount} is effective, and its
point forecasts   are more accurate than the base forecasts at every level of the hierarchy.

The \textit{normal} method does not capture the asymmetry of the base forecasts. Its reconciled point forecasts are  less accurate than those of \textit{probCount}, and its prediction intervals often include  negative values (Figures~\ref{fig:carparts-normal},~\ref{fig:syph-normal}).  

Both \textit{probCount} and \textit{normal} have shorter prediction intervals compared to the base forecast. 
% due to the decreased variance of the reconciled forecasts.
This makes the predictive distribution and the prediction interval more informative, increasing the MIS and the RPS score.
Yet, the prediction intervals of 
both \textit{probCount} and \textit{normal} are sometimes too short.
In future, this could be addressed by modelling the correlation between the base forecasts  using more sophisticated multivariate distributions over counts
\citep{Panagiotelis2012,Inouye2017}.

\subsection{Experiments on  \textit{hospital}}

\begin{figure}[ht!]
	\centering
	\begin{subfigure}[h]{\figwidth}
		\includegraphics[width=\textwidth]{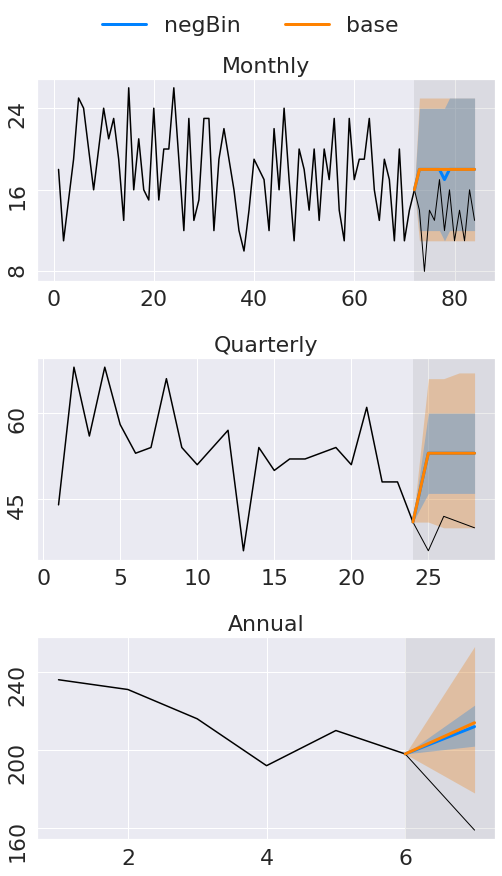}
		\caption{Comparison of forecasts (\textit{base} vs \textit{probCount}) on a time series  from \textbf{hospital}.}
		\label{fig:hospital-base}
	\end{subfigure}
	\begin{subfigure}[h]{\figwidth}
		\includegraphics[width=\textwidth]{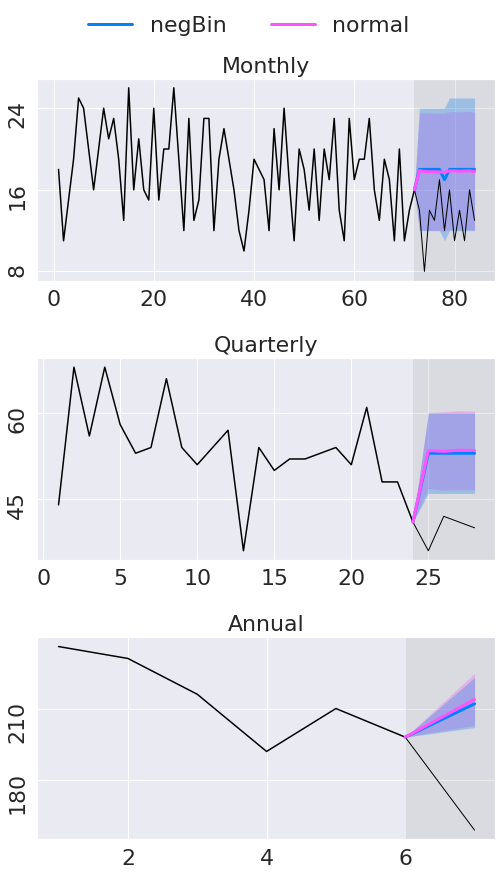}
		\caption{Comparison of forecasts (\textit{normal} vs \textit{probCount}) on a time series  from \textbf{hospital}.}
		\label{fig:hospital-normal}
	\end{subfigure}
	\caption{Examples of reconciliation on a time series from the \textit{hospital} data set.}
	\label{fig:hospital}
\end{figure}

\begin{table}[ht!]
	\centering
	\begin{tabular}{llrrrb}
		\toprule
		\textbf{Skill score on hospital} &  & \multicolumn{3}{c}{Skill score vs \textit{normal}} &  vs base \\
		& & \textit{struc scal} & \textit{truncated} & \textit{probCount} & \textit{probCount}  \\
		\midrule
		ENERGY SCORE  & &    -0.02& 0.00 &  0.00  &   0.03 \\
		MASE  &               &        0.00 &     0.00 &       -0.02 &           0.01 \\
		MIS &               &         -0.02 &      0.00 &        0.00 &           0.03 \\
		RPS &               &        -0.01 &     0.00 &       0.00 &           0.04 \\	 \bottomrule
	\end{tabular}
	\caption{Skill score on time series  from \textit{hospital}, averaged over the entire hierarchy.}
	\label{tab:hospital}
\end{table}

All the  time series extracted from the hospital data set are smooth; the values are high enough to yield symmetric prediction intervals. The samples of the base forecast are well fit  both by the negative binomial and by the Gaussian distribution, among which there are little differences.
Thus the reconciliation methods
become  practically equivalent  (Figure~\ref{fig:hospital-normal}), yielding an almost identical performance (Table~\ref{tab:hospital}). 
The utility of  temporal reconciliation is  confirmed by the positive skill scores  compared to the base forecasts.

\section{Conclusions}
We have shown that virtual evidence, a method originally developed for conditioning probabilistic graphical models on uncertain evidence, can be used to perform probabilistic reconciliation in a principled fashion.
Our method  can reconcile real-valued and  count time series; we focus however on
the latter case, which is especially important since there are currently no methods for reconciling count time series.

The most important result of this paper is that our approach consistently provides a major improvement, 
compared to Gaussian probabilistic reconciliation, in the  reconciliation of  intermittent time series, which  are notoriously hard to forecast. 
Future research directions include  modelling the correlation between the base forecasts and developing a faster  sampling approach in order to reconcile large hierarchies.

\newpage
\appendix
\section{Reconciliation of  a 4-2-1 hierarchy}
\label{sec:421hierarchy}
The 4-2-1 hierarchy of Fig.\ref{fig:hierExample} is reconciled by the code below:
\begin{lstlisting}[language=Python, mathescape=true]
	def reconcile (lambda1, lambda2, lambdaY):
	import pymc3 as pm
	basic_model = pm.Model()
	with basic_model:
	#base forecast of the bottom variables
	Q1 = pm.Poisson ('Q1', mu = lambda1)
	Q2 = pm.Poisson ('Q2', mu = lambda2)
	Q3 = pm.Poisson ('Q3', mu = lambda3)
	Q4 = pm.Poisson ('Q4', mu = lambda4)
	
	#Virtual evidence of the upper variables
	S1 = pm.Poisson ('S1', mu = lambda_S1, observed = S1 + S2)
	S2 = pm.Poisson ('S2', mu = lambda_S2, observed = S3 + S4)
	Y  = pm.Poisson ('Y',  mu = lambda_Y, observed = S1+S2+S3+S4)
	
	#sampling the posterior, i.e., the reconciled distribution
	#for each (s1, s2) computes p(s1) p(s2) p(Y=s1+s2)
	#and eventually normalizes
	trace = pm.sample()
	return trace
\end{lstlisting}

\newpage
\section{Equivalence of definitions~\ref{def:coherence_pan} and \ref{def:coherence} in the continuous case}
\label{sec:equivalentDefs}

\begin{proposition}
	Definitions~\ref{def:coherence_pan} and \ref{def:coherence} are  equivalent in the continuous case.
\end{proposition}
\begin{proof}
	In the continuous case $\mathcal{X}=\mathbb{R}$, thus for any $\mathcal{B} \in \mathcal{F}_{\mathcal{X}^m}$, $s(\mathcal{B}) = \{ s(\mathbf{b}) : \mathbf{b} \in \mathcal{B} \}$. We have that $s(\mathcal{B}) \in \mathcal{F}_{\mathfrak{s}}$ because $\forall \mathbf{b} \in \mathcal{B}$, $s(\mathbf{b}) \in \mathfrak{s}$, by definition, and $\mathcal{F}_{\mathfrak{s}} \subseteq \mathcal{F}_{s(\mathcal{X}^m)}$.
	
	On the other hand, given a element $\mathcal{A} \in \mathcal{F}_\mathfrak{s}$, we can always find a set $\tilde{\mathcal{B}} \in \mathcal{F}_{\mathcal{X}^m}$ such that $s(\tilde{\mathcal{B}}) = \mathcal{A}$. In fact, for any $\mathbf{a} \in \mathcal{A}$ we can write $\mathbf{a} = [\mathbf{a}_{upp},\mathbf{a}_{bot}]$ and $\mathbf{a}_{upp} = S\mathbf{a}_{bot}$ by definition. Thus if we take $\tilde{\mathbf{b}}= \mathbf{a}_{bot}$ we have $s(\tilde{\mathbf{b}}) = \mathbf{a}$. So we have that $\mathcal{A} \in \mathcal{F}_{s(\mathcal{X}^m)}$. 
	
	Since we have showed that the two $\sigma$-algebras $\mathcal{F}_{\mathfrak{s}}$ and $\mathcal{F}_{s(\mathcal{X}^m)}$ are equivalent and the measures $\check{\nu}_\mathfrak{s}$ and $\check{\nu}$ take the same values on all sets of $\mathcal{F}_{\mathcal{X}^m}$ we have that def.~\ref{def:coherence_pan} and def.~\ref{def:coherence} are equivalent in the continuous case. 
\end{proof}

\newpage
\section{Additional fine-grained results on syph and hospital}
\centering
\begin{tabular}{rrrrrb}
	\toprule
	\textbf{Skill score on syph} &  & \multicolumn{3}{c}{vs \textit{normal}} &  vs base \\
	& & \textit{struc scal} & \textit{truncated} & \textit{probCount} & \textit{probCount}  \\
	\midrule
	\textbf{MASE} &      Monthly &    -0.07 & -0.13 &  0.43 &   0.03 \\
	&            2-Monthly &    -0.08 & -0.16 &  0.40 &   0.09 \\
	&            Quarterly &    -0.08 & -0.17 &  0.37 &   0.07 \\
	&            4-Monthly &    -0.07 & -0.18 &  0.35 &   0.07 \\
	&             Biannual &    -0.07 & -0.17 &  0.36 &   0.10 \\
	&               Annual &    -0.02 & -0.19 &  0.20 &   0.01 \\ \rowcolor{lightyellow}
	&     \textit{average} &    -0.06 & -0.17 &  0.35 &   0.06 \\
	\midrule
	\textbf{MIS} &              Monthly &    -0.07 &  0.21 &  0.38 &   0.41 \\
	&            2-Monthly &    -0.16 & -0.09 &  0.19 &   0.33 \\
	&            Quarterly &    -0.16 & -0.18 &  0.13 &   0.29 \\
	&            4-Monthly &    -0.17 & -0.27 &  0.17 &   0.25 \\
	&             Biannual &    -0.11 & -0.25 &  0.21 &   0.14 \\
	&               Annual &    -0.10 & -0.38 &  0.14 &  -0.03 \\ \rowcolor{lightyellow}
	&     \textit{average} &    -0.13 & -0.16 &  0.20 &   0.23 \\
	\midrule
	\textbf{RPS} &              Monthly &    -0.05 & -0.20 &  0.62 &   0.55 \\
	&            2-Monthly &    -0.07 & -0.18 &  0.49 &   0.49 \\
	&            Quarterly &    -0.05 & -0.15 &  0.41 &   0.44 \\
	&            4-Monthly &    -0.05 & -0.16 &  0.40 &   0.42 \\
	&             Biannual &    -0.08 & -0.20 &  0.39 &   0.37 \\
	&               Annual &    -0.04 & -0.17 &  0.21 &   0.22 \\ \rowcolor{lightyellow}
	&     \textit{average} &    -0.06 & -0.18 &  0.42 &   0.41 \\ 
	\bottomrule
\end{tabular}

\begin{tabular}{llrrrb}
	\toprule
	\textbf{Skill score on  hospital} &  & \multicolumn{3}{c}{\textit{normal}} &  vs base \\
	& & \textit{struc scal} & \textit{truncated} & \textit{probCount} & \textit{probCount}  \\
	\midrule
	\textbf{MASE}          & Monthly              &         0.00 &      0.00 &        0.00 &           0.00 \\
	& 2-Monthly            &         0.00 &      0.00 &        0.00 &           0.00 \\
	& Quarterly            &         0.00 &      0.00 &       -0.02 &           0.00 \\
	& 4-Monthly            &         0.00 &      0.00 &       -0.01 &           0.00 \\
	& Biannual             &         0.00 &      0.00 &        0.00 &           0.00 \\
	& Annual               &         0.00 &      0.00 &       -0.07 &           0.04 \\
	\rowcolor{lightyellow} & average              &         0.00 &      0.00 &       -0.02 &           0.01 \\
	\midrule
	\textbf{MIS}           & Monthly              &         0.01 &      0.00 &        0.01 &           0.12 \\
	& 2-Monthly            &         0.00 &      0.00 &        0.00 &           0.15 \\
	& Quarterly            &         0.00 &      0.00 &        0.00 &           0.17 \\
	& 4-Monthly            &        -0.02 &      0.00 &        0.01 &           0.16 \\
	& Biannual             &        -0.04 &      0.00 &        0.00 &       -0.05 \\
	& Annual               &        -0.07 &      0.00 &        0.00 &       -0.38 \\
	\rowcolor{lightyellow} & average              &        -0.02 &      0.00 &        0.00 &           0.03 \\
	\midrule
	\textbf{RPS}           & Monthly              &         0.00 &      0.00 &        0.02 &           0.05 \\
	& 2-Monthly            &         0.00 &      0.00 &        0.00 &           0.09 \\
	& Quarterly            &         0.00 &      0.00 &       -0.01 &           0.09 \\
	& 4-Monthly            &        -0.01 &      0.00 &        0.00 &           0.05 \\
	& Biannual             &        -0.02 &      0.00 &        0.05 &           0.01 \\
	& Annual               &        -0.03 &      0.00 &       -0.08 &          -0.08 \\
	\rowcolor{lightyellow} & average              &        -0.01 &      0.00 &       0.00 &           0.04 \\ \bottomrule
\end{tabular}

\newpage
\subsection*{Proof of Proposition 1}
\begin{proof}
	We prove this by induction over the number of upper time series $n-m$.
	
	Note that if we have one upper time series the two procedures are the same. 
	
	Assume that the two procedures are the same for $n-m-1$ upper time series, i.e. the sequential update after $n-m-1$ upper time series forecasts reconciliations is $\tilde{p}_{n-m-1}(\mathbf{y}) = \mathds{1}_{\mathbf{A}\mathbf{b}}(\mathbf{u}) \dfrac{\sum_{\mathbf{u}^*_{-1}} p_{BU}(\mathbf{u}^*_{-1}, \mathbf{b}) \hat{p}(\mathbf{u}^*_{-1})}{\sum_{\mathbf{u}^*_{-1}} \sum_{\mathbf{b}^*} p_{BU}(\mathbf{u}^*_{-1}, \mathbf{b}) \hat{p}(\mathbf{u}^*_{-1})}$, where we denote by $\mathbf{u}^*_{-1}$ a generic value  for the pmf of the first $n-m-1$ upper forecasts.
	
	The next sequential update is 
	$$\tilde{p}_{n-m}(\mathbf{b}) = \dfrac{\sum_{{u}^*_{n-m}} \tilde{p}_{n-m-1}({u}^*_{n-m}, \mathbf{b}) \hat{p}({u}^*_{n-m})}{\sum_{{u}^*_{n-m}} \sum_{\mathbf{b}^*} \mathds{1}_{\mathbf{A}[n-m,]\mathbf{b}}({u}^*_{n-m})\tilde{p}_{n-m-1}( \mathbf{b}) \hat{p}({u}^*_{n-m})}.$$ 
	We denote the denominator as $Z_{n-m}$, a normalizing constant, and we study just the numerator $\tilde{p}_{n-m}(\mathbf{b})Z_{n-m}$.
	\begin{align}
		\nonumber
		\tilde{p}_{n-m}(\mathbf{b})Z_{n-m} &= \sum_{{u}^*_{n-m}} \tilde{p}_{n-m-1}({u}^*_{n-m}, \mathbf{b}) \hat{p}({u}^*_{n-m}) \\ \nonumber
		&= \sum_{{u}^*_{n-m}} \mathds{1}_{\mathbf{A}[n-m,]\mathbf{b}}({u}^*_{n-m})\tilde{p}_{n-m-1}( \mathbf{b}) \hat{p}({u}^*_{n-m}) \\ \label{eq:lastPn-m}
		&= \sum_{{u}^*_{n-m}} \mathds{1}_{\mathbf{A}[n-m,]\mathbf{b}}({u}^*_{n-m}) \dfrac{\sum_{\mathbf{u}^*_{-1}} p_{BU}(\mathbf{u}^*_{-1}, \mathbf{b}) \hat{p}(\mathbf{u}^*_{-1})}{Z_{n-m-1}}\hat{p}({u}^*_{n-m}),
	\end{align}
	where we denote by $Z_{n-m-1}$ the normalizing constant for the reconciled distribution after $n-m-1$ updates. Note that $p_{BU}(\mathbf{u}^*_{-1}, \mathbf{b}) = \mathds{1}_{\mathbf{A}[-1,]\mathbf{b}}(\mathbf{u}^*_{-1}) \hat{p}(\mathbf{b})$, where $\mathbf{A}[-1,]$ indicates all but the last row of the matrix $\mathbf{A}$. Moreover 
	$$\mathds{1}_{\mathbf{A}[-1,]\mathbf{b}}(\mathbf{u}^*_{-1}) \mathds{1}_{\mathbf{A}[n-m,]\mathbf{b}}({u}^*_{n-m}) = \mathds{1}_{\mathbf{A}\mathbf{b}}({u}^*),$$ 
	if $\mathbf{u}^* = [ \mathbf{u}^*_{-1} \ u^*_{n-m}]$. Therefore the double sum in equation \eqref{eq:lastPn-m} can be simplified as
	\begin{align*}
		\tilde{p}_{n-m}(\mathbf{b})Z_{n-m} &= \dfrac{\sum_{\mathbf{u}^* = [\mathbf{u}^*_{-1} \ {u}^*_{n-m}]} \mathds{1}_{\mathbf{A}\mathbf{b}}(\mathbf{u}^*) \hat{p}(\mathbf{b}) \hat{p}(\mathbf{u}^*_{-1})\hat{p}({u}^*_{n-m})}{Z_{n-m-1}} \\
		&= \dfrac{\sum_{\mathbf{u}^* = [\mathbf{u}^*_{-1} \ {u}^*_{n-m}]} p_{BU}(\mathbf{u}^*, \mathbf{b}) \hat{p}(\mathbf{u}^*_{-1})\hat{p}({u}^*_{n-m})}{Z_{n-m-1}}
	\end{align*}
	Since the upper time series forecast is conditionally independent we have $\hat{p}(\mathbf{u}^*) = \hat{p}(\mathbf{u}^*_{-1})\hat{p}({u}^*_{n-m})$ and thus we obtain the result. 
	
\end{proof}

\end{document}